\documentclass[12pt, reqno]{amsart}

\usepackage{graphicx}
\usepackage{epic}
\usepackage{amsmath}
\usepackage{amsfonts}
\usepackage{amssymb}
\usepackage{amsthm}
\usepackage{eucal}
\usepackage[ansinew]{inputenc}

\newtheorem{theorem}{Theorem}[section]

\newtheorem{corollary}[theorem]{Corollary}
\newtheorem{proposition}[theorem]{Proposition}
\newtheorem{remark}[theorem]{Remark}

\numberwithin{equation}{section}

\begin{document}

\title[A centerless representation of the Virasoro algebra]{A centerless representation of the Virasoro algebra
associated with the unitary circular ensemble}

\author[L.~Haine]{Luc Haine}
\address{L.H., Department of Mathematics, Université catholique de
Louvain, Chemin du Cyclotron 2, 1348 Louvain-la-Neuve, Belgium}
\email{luc.haine@uclouvain.be}

\author[D.~Vanderstichelen]{Didier Vanderstichelen}
\address{D.V., Department of Mathematics, Université catholique de
Louvain, Chemin du Cyclotron 2, 1348 Louvain-la-Neuve, Belgium}
\email{didier.vanderstichelen@uclouvain.be}

\date{January 18 2010}
\thanks{The authors acknowledge the partial support of the Belgian Interuniversity Attraction Pole P06/02. The second author is a Research Fellow at FNRS, Belgium.}

\begin{abstract} We consider the 2-dimensional Toda lattice tau functions $\tau_n(t,s;\eta,\theta)$ deforming the probabilities
$\tau_n(\eta,\theta)$ that a randomly chosen matrix from the unitary group $U(n)$, for the Haar measure, has no eigenvalues
within an arc $(\eta,\theta)$ of the unit circle. We show that these tau functions satisfy a centerless Virasoro algebra of constraints,
with a boundary part in the sense of Adler, Shiota and van Moerbeke. As an application, we obtain a new derivation of a differential
equation due to Tracy and Widom, satisfied by these probabilities, linking it to the Painlevé VI equation.
\end{abstract}
\maketitle
\section{Introduction} 
Consider the group $U(n)$ of $n\times n$ unitary matrices, with the normalized Haar measure as a probability measure. The Weyl integral
formula gives the induced density distribution on the eigenvalues of the matrices on the unit circle in the complex plane, and is given by
\begin{equation*}
\frac{1}{n!}\vert\Delta_n(z)\vert^2 \prod_{k=1}^{n}\frac{\mbox{d}z_k}{2\pi iz_k};\quad z_k=e^{i\varphi_k}\;
\mbox{and}\;\Delta_n(z)=\prod_{1\leq k<l\leq n} ( z_k-z_l).
\end{equation*}
Thus, for $\eta, \theta \in ]-\pi,\pi[$, with $\eta\leq \theta$, the probability that a randomly chosen matrix from $U(n)$ has
no eigenvalues within an arc of circle $(\eta, \theta)=\{z\in S^1\vert \eta<\mbox{arg}(z)< \theta\}$ is given by
\begin{equation*}
\tau_n(\eta,\theta)=\frac{1}{(2\pi)^n n!}\int_{\theta}^{2\pi+\eta}\ldots \int_{\theta}^{2\pi+\eta}\prod_{1\leq k<l\leq n}\vert e^{i\varphi_k}-e^{i\varphi_l}\vert^2 \mbox{d}\varphi_1\ldots\mbox{d}\varphi_n.
\end{equation*}
Obviously, this probability depends only on the length $\theta-\eta$. All of this is well known and we refer the
reader to Mehta \cite{Me} for details. We shall denote by
\begin{equation} \label{gaprob}
R(\theta)=-\frac{1}{2}\frac{\mbox{d}}{\mbox{d}\theta}\log\tau_n(-\theta,\theta),
\end{equation}
the logarithmic derivative of the probability that an arc of circle of length $2\theta$ contains no eigenvalues of a randomly chosen unitary matrix.

The starting motivation for our work was to understand a differential equation satisfied by the function $R(\theta)$
\begin{gather}
R(\theta)^2+2\sin\theta\,\cos\theta\,R(\theta)R'(\theta)+\sin^2\theta\,R'(\theta)^2\nonumber\\
=\frac{1}{2}\Big(\frac{1}{4}\sin^2\theta\,\frac{R''(\theta)^2}{R'(\theta)}
+\sin\theta\,\cos\theta\,R''(\theta)+\big(\cos^2\theta+n^2\sin^2\theta\big)\,R'(\theta)\Big),\label{TW}
\end{gather}
obtained by Tracy and Widom in \cite{TW}, from the point of view of the Adler-Shiota-van Moerbeke \cite{ASVM} approach, in terms of Virasoro constraints. Introducing the 2-Toda time-dependent tau functions
\begin{equation} \label{tauint}
\tau_n(t,s;\eta,\theta)=\frac{1}{n!}\int_{[\theta, 2\pi+\eta]^n} \vert\Delta_n(z)\vert^2\,\prod_{k=1}^{n}\Big(e^{\sum_{j=1}^\infty(t_jz_k^j+s_jz_k^{-j})}\;\frac{\mbox{d}z_k}{2\pi iz_k}\Big),
\end{equation}
with $z_k=e^{i\varphi_k}$, deforming the probabilities $\tau_n(\eta,\theta)=\tau_n(0,0;\eta,\theta)$, we discover that they satisfy a set of Virasoro constraints
indexed by \emph{all} integers, decoupling into a boundary-part and a time-part
\begin{equation*}
\frac{1}{i}\Big(e^{ik\theta}\frac{\partial}{\partial \theta}+e^{ik\eta}\frac{\partial}{\partial \eta}\Big)\tau_n(t,s;\eta,\theta)=L_k^{(n)}\tau_n(t,s;\eta,\theta),
\;k\in\mathbb{Z},\;i=\sqrt{-1},
\end{equation*}
with the time-dependent operators $L_k^{(n)}$ providing a centerless representation of the \emph{full} Virasoro algebra, see Section 2
(Theorem 2.2) for a precise statement and the proof of the result.

In their study of Painlevé equations satisfied (as functions of $x$) by integrals of Gessel type  $E_{U(n)}e^{x\;\emph{tr}(M+\overline{M})}$, where the expectation $E_{U(n)}$ refers to integration with respect to the Haar measure over the whole of $U(n)$, Adler and van Moerbeke \cite{AVM2} found the $sl_2$ subalgebra corresponding to $k=-1,0,1$, without boundary terms. The appearance of boundary terms and of a \emph{full} centerless Virasoro algebra is to the best of our knowledge new. From this result, it is easy to obtain equation \eqref{TW}, using the algorithmic method of \cite{ASVM}. Finally, similarly to a result by the first author and Semengue \cite{HS} on the Jacobi polynomial ensemble, we show that $R(\theta)$ is the restriction to the unit circle of a function $r(z)$ defined in the complex plane, so that $\sigma(z)=-i(z-1)r(z)-n^2z/4$ satisfies a special case of the Okamoto-Jimbo-Miwa form of the Painlevé VI equation. This will be explained in Section 3 of the paper.
\section{A centerless representation of the Virasoro algebra} 
The proof of the Virasoro constraints satisfied by the integral \eqref{tauint} is a non-trivial adaptation of the self-similarity argument exploited in the case of the Gaussian ensembles, based on the invariance of the integrals with respect to translations, see \cite{AVM1} and references therein. Here, we replace translations by appropriate rotations. More precisely, setting
\begin{equation} \label{integrand}
\mbox{d}I_n(t,s,z)=\vert\Delta_n(z)\vert^2\,\prod_{\alpha=1}^{n}\Big(e^{\sum_{j=1}^\infty(t_jz_\alpha^j+s_jz_\alpha^{-j})}
\;\frac{\mbox{d}z_\alpha}{2\pi iz_\alpha}\Big),
\end{equation}
with $z_\alpha=e^{i\varphi_\alpha}$ and $\vert\Delta_n(z)\vert^2=\prod_{1\leq \alpha<\beta\leq n}\vert z_\alpha-z_\beta\vert^2$,
we have the fundamental next proposition.
\begin{proposition} The following variational formulas hold
\begin{align}
\frac{\emph{d}}{\emph{d}\varepsilon}\;\emph{d}I_n \big(z_\alpha\mapsto z_\alpha e^{\varepsilon (z_\alpha^k -z_\alpha^{-k})}\big)\big|_ {\varepsilon=0}&=\big(L_k^{(n)}-L_{-k}^{(n)}\big)\;\emph{d}I_n,\label{var1}\\
\frac{\emph{d}}{\emph{d}\varepsilon}\;\emph{d}I_n \big(z_\alpha\mapsto z_\alpha e^{i\varepsilon (z_\alpha^k +z_\alpha^{-k})}\big)\big|_ {\varepsilon=0}&=i\big(L_k^{(n)}+L_{-k}^{(n)}\big)\;\emph{d}I_n\label{var2},
\end{align}
for all $k\geq 0$, with
\begin{align}
L_k^{(n)}=&\sum_{j=1}^{k-1}\frac{\partial^2}{\partial t_{j}\partial t_{k-j}}+n\frac{\partial}{\partial t_k}+\sum_{j=1}^{\infty}jt_j\frac{\partial}{\partial t_{j+k}}\notag\\&-\sum_{j=k+1}^{\infty}js_j\frac{\partial}{\partial s_{j-k}}-\sum_{j=1}^{k-1}js_j\frac{\partial}{\partial t_{k-j}}-nks_k,\quad k\geq 1,\label{virplus}\\
L_0^{(n)}=&\sum_{j=1}^{\infty}jt_j\frac{\partial}{\partial t_{j}}-\sum_{j=1}^{\infty}js_j\frac{\partial}{\partial s_{j}},\label{virzero}\\
L_{-k}^{(n)}=&-\sum_{j=1}^{k-1}\frac{\partial^2}{\partial s_j\partial s_{k-j}}-n\frac{\partial}{\partial s_k}-\sum_{j=1}^{\infty}js_j\frac{\partial}{\partial s_{j+k}}\notag\\&+\sum_{j=k+1}^{\infty}jt_j\frac{\partial}{\partial t_{j-k}}+\sum_{j=1}^{k-1}jt_j\frac{\partial}{\partial s_{k-j}}+nkt_k,\quad k\geq 1\label{virminus}.
\end{align}
\end{proposition}
\begin{proof} We shall only give the proof of \eqref{var1}, the proof of \eqref{var2} is similar. Upon setting
\begin{equation*}
E=\prod_{\alpha=1}^n e^{\sum_{j=1}^\infty(t_jz_\alpha^j+s_jz_\alpha^{-j})},
\end{equation*}
the following four relations hold, for $k\geq 0$,
\begin{align}
\Big(\frac{\partial}{\partial t_k}+n\delta_{k,0}\Big)E&=\Big(\sum_{\alpha=1}^n z_\alpha^k\Big)E\notag\\
\Big(\frac{\partial}{\partial s_k}+n\delta_{k,0}\Big)E&=\Big(\sum_{\alpha=1}^n z_\alpha^{-k}\Big)E, \label{E1}\\
\Bigg(\frac{1}{2}\sum_{\substack{i+j=k\\i,j>0}}\frac{\partial^2}{\partial t_i\partial t_j}-\frac{n}{2}\delta_{k,0}\Bigg)E&=
\Bigg(\sum_{\substack{1\leq \alpha<\beta\leq n\\i+j=k\\i,j>0}} z_\alpha^i z_\beta^j+\frac{k-1}{2}\sum_{\alpha=1}^n z_\alpha^k\Bigg)E\notag\\
\Bigg(\frac{1}{2}\sum_{\substack{i+j=k\\i,j>0}}\frac{\partial^2}{\partial s_i\partial s_j}-\frac{n}{2}\delta_{k,0}\Bigg)E&=
\Bigg(\sum_{\substack{1\leq \alpha<\beta\leq n\\i+j=k\\i,j>0}} z_\alpha^{-i} z_\beta^{-j}+\frac{k-1}{2}\sum_{\alpha=1}^n z_\alpha^{-k}\Bigg)E.
\label{E2}
\end{align}

We split the computation into four contributions, corresponding to various factors in \eqref{integrand}.\\

\emph{Contribution 1}: For $k>0$, we have
\begin{align*}
&\frac{\partial}{\partial \varepsilon} \big\vert\Delta_n\big(ze^{\varepsilon(z^k-z^{-k})}\big)\big\vert^2\Big|_{\varepsilon=0}
\\&=\vert\Delta_n(z)\vert^2 \sum_{1\leq \alpha<\beta\leq n}\frac{(z_\alpha+z_\beta)(z_\alpha^k-z_\beta^k-(z_\alpha^{-k}-z_\beta^{-k}))}{z_\alpha-z_\beta}
\\&=\vert\Delta_n(z)\vert^2 \sum_{1\leq \alpha<\beta\leq n}(z_\alpha+z_\beta)\Big(\sum_{i=0}^{k-1}z_\alpha^iz_\beta^{k-1-i}+
\sum_{i=0}^{k-1}z_\alpha^{-i-1}z_\beta^{i-k}\Big)
\\&=\vert\Delta_n(z)\vert^2 E^{-1}\Bigg[2\sum_{\substack{1\leq \alpha<\beta\leq n\\i+j=k\\i,j>0}} (z_\alpha^i z_\beta^j+z_\alpha^{-i} z_\beta^{-j})+(n-1)\sum_{\alpha=1}^n(z_\alpha^k+z_\alpha^{-k})\Bigg]E.
\end{align*}
Using the four relations \eqref{E1} and \eqref{E2}, we obtain
\begin{multline}
\frac{\partial}{\partial \varepsilon} \big\vert\Delta_n\big(ze^{\varepsilon(z^k-z^{-k})}\big)\big\vert^2\Big|_{\varepsilon=0}
 =2\vert\Delta_n(z)\vert^2 E^{-1}\Bigg[\frac{1}{2}\sum_{\substack{i+j=k\\i,j>0}}\frac{\partial^2}{\partial t_i\partial t_j}
\\+\frac{1}{2}\sum_{\substack{i+j=k\\i,j>0}}\frac{\partial^2}{\partial s_i\partial s_j}+\frac{n-k}{2}\frac{\partial}{\partial t_k}+
\frac{n-k}{2}\frac{\partial}{\partial s_k}\Bigg]E,\label{c1}
\end{multline}
which is also trivially satisfied for $k=0$.\\

\emph{Contribution 2}: For $k\geq 0$, using the relations \eqref{E1}, we have
\begin{align}
&\frac{\partial}{\partial \varepsilon}\prod_{\alpha=1}^n \mbox{d}\big(z_\alpha e^{\varepsilon(z_\alpha^k-z_\alpha^{-k})}\big)\Big|_{\varepsilon=0}
\notag\\&=E^{-1}\sum_{\alpha=1}^n\big((k+1)z_\alpha^k+(k-1)z_\alpha^{-k}\big)E\prod_{\alpha=1}^n\mbox{d}z_\alpha
\notag\\&=E^{-1}\Big[(k+1)\frac{\partial}{\partial t_k}+(k-1)\frac{\partial}{\partial s_k}\Big]E\prod_{\alpha=1}^n\mbox{d}z_\alpha.
\label{c2}
\end{align}

\emph{Contribution 3}: For $k\geq 0$, using the relations \eqref{E1}, we have
\begin{align}
&\frac{\partial}{\partial \varepsilon} \prod_{\alpha=1}^n e^{\sum_{j=1}^\infty
\Big(t_j\big(z_\alpha e^{\varepsilon (z_\alpha^k-z_\alpha^{-k})}\big)^j
+s_j\big(z_\alpha e^{\varepsilon (z_\alpha^k-z_\alpha^{-k})}\big)^{-j}\Big)}\Big|_{\varepsilon=0}
\notag\\&= \sum_{\alpha=1}^n \Big[\sum_{j=1}^\infty jt_jz_\alpha^j(z_\alpha^k-z_\alpha^{-k})-\sum_{j=1}^{\infty}js_jz_{\alpha}^{-j}(z_\alpha^k-z_\alpha^{-k})\Big]E
\notag\\&=\Bigg[\sum_{j=1}^\infty jt_j\sum_{\alpha=1}^nz_\alpha^{j+k}-\sum_{j=1}^{k-1}jt_j\sum_{\alpha=1}^nz_\alpha^{j-k}-
\sum_{j=k}^\infty jt_j\sum_{\alpha=1}^nz_\alpha^{j-k}
\notag\\&\qquad -\sum_{j=1}^{k-1}js_j\sum_{\alpha=1}^nz_\alpha^{k-j}-\sum_{j=k}^\infty js_j\sum_{\alpha=1}^nz_\alpha^{k-j}
+\sum_{j=1}^\infty js_j\sum_{\alpha=1}^nz_\alpha^{-k-j}\Bigg]E
\notag\\&=\Bigg[\sum_{j=1}^\infty jt_j\frac{\partial}{\partial t_{k+j}}-\sum_{j=1}^{k-1}jt_j\frac{\partial}{\partial s_{k-j}}
-\sum_{j=k+1}^\infty jt_j\frac{\partial}{\partial t_{j-k}}-nkt_k
\notag\\&\qquad -\sum_{j=1}^{k-1}js_j\frac{\partial}{\partial t_{k-j}}-\sum_{j=k+1}^\infty js_j\frac{\partial}{\partial s_{j-k}}-nks_k
+\sum_{j=1}^\infty js_j\frac{\partial}{\partial s_{k+j}}\Bigg]E.\label{c3}
\end{align}

\emph{Contribution 4}: For $k\geq 0$, using the relations \eqref{E1}, we have
\begin{align}
\frac{\partial}{\partial \varepsilon} \prod_{\alpha=1}^n \frac{1}{2\pi i z_\alpha e^{\varepsilon(z_\alpha^k-z_\alpha^{-k})}}\Big|_{\varepsilon=0}
&=E^{-1}\Big[-\sum_{\alpha=1}^nz_\alpha^k+\sum_{\alpha=1}^nz_\alpha^{-k}\Big]E\prod_{\alpha=1}^n\frac{1}{2\pi i z_\alpha}
\notag\\&=E^{-1}\Big[-\frac{\partial}{\partial t_k}+\frac{\partial}{\partial s_k}\Big]E\prod_{\alpha=1}^n\frac{1}{2\pi i z_\alpha}.
\label{c4}
\end{align}

Adding up \eqref{c1}, \eqref{c2}, \eqref{c3} and \eqref{c4} gives \eqref{var1}. This concludes the proof of Proposition 2.1.
\end{proof}
We are now able to state our main result.
\begin{theorem} \emph{(i)} The tau functions \footnote{See the beginning of Section 3, for a justification of the terminology.} $\tau_n(t,s;\eta,\theta), n\geq 1$, defined in \eqref{tauint}, satisfy
\begin{equation} \label{virb}
\mathcal{B}_k(\eta,\theta)\tau_n(t,s;\eta,\theta)=L_k^{(n)}\tau_n(t,s;\eta,\theta),\quad k\in\mathbb{Z},
\end{equation}
with $L_k^{(n)},k\in\mathbb{Z}$, defined as in \eqref{virplus}, \eqref{virzero}, \eqref{virminus}, and
\begin{equation}\label{b}
\mathcal{B}_k(\eta,\theta)=\frac{1}{i}\Big(e^{ik\theta}\frac{\partial}{\partial \theta}+e^{ik\eta}\frac{\partial}{\partial \eta}\Big);
\quad i=\sqrt{-1}.
\end{equation}

\emph{(ii)} The operators $L_k^{(n)},k\in\mathbb{Z}$, satisfy the commutation relations of the centerless Virasoro algebra, that is
\begin{equation}\label{virc}
\big[L_k^{(n)},L_l^{(n)}\big]=(k-l)L_{k+l}^{(n)},\quad k,l\in\mathbb{Z}.
\end{equation}
\end{theorem}
\begin{proof} (i) Denoting $z_\alpha=e^{i\varphi_\alpha}$, the change of variable $z_\alpha\mapsto z_\alpha e^{\varepsilon (z_\alpha^k -z_\alpha^{-k})}$ in the integral \eqref{tauint} gives the following transformation on the angle
$\varphi_\alpha\mapsto \varphi_\alpha+2\varepsilon \sin(k\varphi_\alpha)$, inducing a change in the limits of integration given by the inverse map
\begin{equation} \label{lim1}
\varphi_\alpha\mapsto \varphi_\alpha-2\varepsilon\sin(k\varphi_\alpha)+O(\varepsilon^2),
\end{equation}
for $\varepsilon$ small enough. Making the change of variable in the integral \eqref{tauint}, with the corresponding change in the limits of integration, leaves it invariant. Thus, by differentiating the result with respect to $\varepsilon$ and evaluating it at $\varepsilon=0$, using the chain rule together with \eqref{var1} and \eqref{lim1}, we obtain
\begin{equation}\label{dint1}
0=\Big(-2\sin(k\theta)\frac{\partial}{\partial\theta}-2\sin(k\eta)\frac{\partial}{\partial\eta}
+L_k^{(n)}-L_{-k}^{(n)}\Big)\tau_n(t,s;\eta,\theta).
\end{equation}
Similarly, the change of variable $z_\alpha\mapsto z_\alpha e^{i\varepsilon (z_\alpha^k +z_\alpha^{-k})}$ corresponds to the transformation
$\varphi_\alpha\mapsto \varphi_\alpha+2\varepsilon \cos(k\varphi_\alpha)$, with inverse
\begin{equation*}
\varphi_\alpha\mapsto \varphi_\alpha-2\varepsilon \cos(k\varphi_\alpha)+O(\varepsilon^2),
\end{equation*}
which, using \eqref{var2}, leads to
\begin{equation}\label{dint2}
0=\Big(-\frac{2}{i}\cos(k\theta)\frac{\partial}{\partial\theta}-\frac{2}{i}\cos(k\eta)\frac{\partial}{\partial\eta}
+L_k^{(n)}+L_{-k}^{(n)}\Big)\tau_n(t,s;\eta,\theta).
\end{equation}
Adding and subtracting \eqref{dint1} and \eqref{dint2} gives the constraints \eqref{virb}, with $\mathcal{B}_k(\eta,\theta)$ defined as in \eqref{b}.

(ii) Consider the complex Lie algebra $\mathcal{A}$ given by the direct sum of two commuting copies of the Heisenberg algebra with bases $\{\hbar_a,a_j|j\in\mathbb{Z}\}$ and $\{\hbar_b,b_j|j\in\mathbb{Z}\}$ and defining commutation relations
\begin{gather}
[\hbar_a,a_j]=0\quad,\quad [a_j,a_k]=j\delta_{j,-k}\hbar_a,\notag\\
[\hbar_b,b_j]=0\quad,\quad [b_j,b_k]=j\delta_{j,-k}\hbar_b,\label{Heisenberg algebra}\\
[\hbar_a,\hbar_b]=0\quad,\quad [a_j,b_k]=0\quad,\quad [\hbar_a,b_j]=0\quad,\quad [\hbar_b,a_j]=0,\notag
\end{gather}
with $j,k\in\mathbb{Z}$. Let $\mathcal{B}$ be the space of formal power series in the variables $t_1,t_2,\dots$ and $s_1,s_2,\dots$, and consider the following representation of $\mathcal{A}$ in $\mathcal{B}$ :
\begin{gather}
a_j=\frac{\partial}{\partial t_j}\quad,\quad a_{-j}=jt_j \quad,\quad b_j=\frac{\partial}{\partial s_j}\quad,\quad b_{-j}=js_j,\notag\\
a_0=b_0=\mu\quad,\quad \hbar_a=\hbar_b=1,\label{oscillator representation Heisenberg}
\end{gather}
for $j>0$, and $\mu\in\mathbb{C}$. Define the operators
\begin{equation*}
A_k^{(n)}=\frac{1}{2}\sum_{j\in\mathbb{Z}}:a_{-j}a_{j+k}:\quad,\quad B_k^{(n)}=\frac{1}{2}\sum_{j\in\mathbb{Z}}:b_{-j}b_{j+k}:,
\end{equation*}
where $k\in\mathbb{Z}$, $a_j,b_j$ are as in \eqref{oscillator representation Heisenberg} with $\mu=n$, and where the colons indicate normal ordering, defined by
\begin{equation*}
:a_ja_k:=\left\{\begin{array}{ll} a_ja_k\quad\text{if $j\leq k$},\\ a_ka_j\quad\text{if $j>k$},\end{array}\right.
\end{equation*}
and a similar definition for $:b_jb_k:$, obtained by changing the $a$'s in $b$'s in the former. Using these notations, we can rewrite \eqref{virplus}, \eqref{virzero} and \eqref{virminus} as follows
\begin{align}
&L_k^{(n)}=A_k^{(n)}-B_{-k}^{(n)}+\frac{1}{2}\sum_{j=1}^{k-1}(a_j-b_{-j})(a_{k-j}-b_{j-k}),\quad k\geq 1\notag\\
&L_0^{(n)}=A_0^{(n)}-B_0^{(n)},\notag\\
&L_{-k}^{(n)}=A_{-k}^{(n)}-B_k^{(n)}-\frac{1}{2}\sum_{j=1}^{k-1}(a_{-j}-b_{j})(a_{j-k}-b_{k-j}), \quad k\geq 1 .\notag
\end{align}
As shown in \cite{KR} (see Lecture 2) the operators $A_k^{(n)}$, $k\in\mathbb{Z}$, provide a representation of the Virasoro algebra in $\mathcal{B}$ with central charge $c=1$, that is
\begin{equation} \label{Virasoro A}
[A_k^{(n)},A_l^{(n)}]=(k-l)A_{k+l}^{(n)}+\delta_{k,-l}\frac{k^3-k}{12},
\end{equation}
for $k,l\in\mathbb{Z}$. Similarly, the operators $B_k^{(n)}$ satisfy the commutation relations
\begin{equation} \label{Virasoro B}
[B_k^{(n)},B_l^{(n)}]=(k-l)B_{k+l}^{(n)}+\delta_{k,-l}\frac{k^3-k}{12},
\end{equation}
for $k,l\in\mathbb{Z}$. Furthermore we have for $k,l\in\mathbb{Z}$
\begin{gather}
[a_k,A_l^{(n)}]=ka_{k+l}\quad,\quad[b_k,B_l^{(n)}]=kb_{k+l},\notag\\
[a_k,B_l^{(n)}]=0\quad,\quad[b_k,A_l^{(n)}]=0.\label{commutator Heisenberg and Virasoro operators}
\end{gather}

Let us now establish the commutation relations \eqref{virc}. We give the proof for $k,l\geq 0$, the other cases being similar. As $[A_{i}^{(n)},B_{j}^{(n)}]=0$, $i,j\in\mathbb{Z}$, we have using \eqref{Heisenberg algebra}, \eqref{Virasoro A}, \eqref{Virasoro B} and \eqref{commutator Heisenberg and Virasoro operators}
\begin{multline*}
[L_k^{(n)},L_l^{(n)}]=(k-l)\big(A_{k+l}^{(n)}-B_{-k-l}^{(n)}\big)
-\frac{1}{2}\sum_{j=1}^{l-1}j(a_{j+k}-b_{-j-k})(a_{l-j}-b_{j-l})\\
-\frac{1}{2}\sum_{j=1}^{l-1}(l-j)(a_{j}-b_{-j})(a_{k+l-j}-b_{j-k-l})
+\frac{1}{2}\sum_{j=1}^{k-1}j(a_{j+l}-b_{-j-l})(a_{k-j}-b_{j-k})\\
+\frac{1}{2}\sum_{j=1}^{k-1}(k-j)(a_{j}-b_{-j})(a_{k+l-j}-b_{j-k-l}).
\end{multline*}
Relabeling the indices in the sums, we have
\begin{align}
[L_k^{(n)},L_l^{(n)}]=&(k-l)\big(A_{k+l}^{(n)}-B_{-k-l}^{(n)}\big)\notag\\
&-\frac{1}{2}\sum_{j=k+1}^{k+l-1}(j-k)(a_{j}-b_{-j})(a_{k+l-j}-b_{j-k-l})\notag\\
&-\frac{1}{2}\sum_{j=1}^{l-1}(l-j)(a_{j}-b_{-j})(a_{k+l-j}-b_{j-k-l})\notag\\
&+\frac{1}{2}\sum_{j=l+1}^{k+l-1}(j-l)(a_{j}-b_{-j})(a_{k+l-j}-b_{j-k-l})\notag\\
&+\frac{1}{2}\sum_{j=1}^{k-1}(k-j)(a_{j}-b_{-j})(a_{k+l-j}-b_{j-k-l})\notag\\
=&(k-l)L_{k+l}^{(n)}.\notag
\end{align}
This concludes the proof of Theorem 2.2.
\end{proof}
\section{The unitary circular ensemble and the Painlevé VI equation} 
It is well known, see for instance \cite{Me}, that the integral $\tau_n(t,s;\eta,\theta)$ defined in \eqref{tauint} can be represented as a Toeplitz determinant
\begin{equation} \label{taudet}
\tau_n(t,s;\eta,\theta)=\mbox{det}\big(\mu_{k-l}(t,s;\eta,\theta)\big)_{0\leq k,l\leq n-1},
\end{equation}
with
\begin{equation*}
\mu_k(t,s;\eta,\theta)=\int_{\theta}^{2\pi+\eta} z^k \Big(e^{\sum_{j=1}^\infty(t_jz ^j+s_jz ^{-j})}\;\frac{\mbox{d}z}{2\pi iz}\Big);\
z=e^{i\varphi}, k\in\mathbb{Z}.
\end{equation*}
A nice consequence of this representation is that $\tau_n(t,s;\eta,\theta)$ is a tau function of a reduction of the 2-Toda lattice hierarchy, that was called the Toeplitz hierarchy in  \cite{AVM2}. Therefore, as with any 2-Toda tau function (see \cite{UT}), it satisfies the KP equation in the $t=(t_1,t_2,\ldots)$ (or $s=(s_1,s_2,\ldots)$) variables separately
\begin{equation}
\Big(\frac{\partial^4}{\partial t_1^4} + 3\frac{\partial^2}{\partial t_2^2}-4\frac{\partial^2}{\partial t_1\partial t_3}\Big)\log\tau_n+6 \Big(\frac{\partial^2}{\partial t_1^2}\log\tau_n\Big)^2=0.\label{KP}
\end{equation}

As announced in the introduction, in this section, using the method of \cite{ASVM}, we establish the following result.
\begin{theorem} The Virasoro constraints \eqref{virb}, combined with the KP equation \eqref{KP} in the $t$ variables (or the KP equation in the $s$ variables), imply that the function $R(\theta)$ defined in \eqref{gaprob} satisfies \eqref{TW}.
\end{theorem}
\begin{proof}
Remembering the definition of $L_0^{(n)}$ in \eqref{virzero}, the Virasoro constraint in \eqref{virb} for $k=0$,
evaluated along the locus $t=s=0$, gives
\begin{equation}
\frac{\partial \log\tau_n(t,s;\eta,\theta)}{\partial \theta}\Bigg|_{t=s=0}=-\frac{\partial \log\tau_n(t,s;\eta,\theta)}{\partial\eta}\Bigg|_{t=s=0},\label{partial theta = - partial eta}
\end{equation}
which is a reformulation of the fact that the gap probability $\tau_n(0,0;\eta,\theta)$ only depends on the length $\theta-\eta$.

Define the operator $\mathcal{D}=\frac{\partial}{\partial\theta}-\frac{\partial}{\partial\eta}$ and put for a fixed $n$
\begin{align}
&f(t,s;\eta,\theta)=\log\tau_n(t,s;\eta,\theta),\notag\\
&g(\eta,\theta)=-\frac{1}{2}\mathcal{D}\log\tau_n(t,s;\eta,\theta)\big|_{t=s=0}.\label{g}
\end{align}
Notice that for $k\geq 0$
\begin{equation}
\mathcal{D}^k\log\tau_n(t,s;\eta,\theta)\big|_{\substack{t=s=0 \\ \eta=-\theta}}=\frac{\mbox{d}^k}{\mbox{d}\theta^k}\log\tau_n(t,s;-\theta,\theta)\big|_{t=s=0}.\notag
\end{equation}
Clearly, from the definition of $R(\theta)$ in \eqref{gaprob}, we have
\begin{equation*}
R(\theta)=g(-\theta,\theta)=-\frac{1}{2}\frac{\mbox{d}}{\mbox{d}\theta}\log\tau_n(t,s;-\theta,\theta)\big|_{t=s=0}.
\end{equation*}

Remembering the definition of $L_k^{(n)}$ in \eqref{virplus}, the constraints in \eqref{virb} for $k=1,2$, evaluated at $s=(s_1,s_2,s_3,\dots)=(0,0,0,\dots)$, can be written
\begin{align}
&\mathcal{B}_1(\eta,\theta)f\Big|_{s=0}=\sum_{j\geq1}jt_j\,\frac{\partial f}{\partial t_{j+1}}\Bigg|_{s=0}+n\,\frac{\partial f}{\partial t_1}\Bigg|_{s=0},\label{Virasoro bis k=1}\\
&\mathcal{B}_2(\eta,\theta)f\Big|_{s=0}=\sum_{j\geq1}jt_j\,\frac{\partial f}{\partial t_{j+2}}\Bigg|_{s=0}+\frac{\partial^2 f}{\partial t_1^2}\Bigg|_{s=0}+\Big(\frac{\partial f}{\partial t_1}\Big)^2\Bigg|_{s=0}+n\,\frac{\partial f}{\partial t_2}\Bigg|_{s=0}.\label{Virasoro bis k=2}
\end{align}
Using \eqref{partial theta = - partial eta} and the definition of $g(\eta,\theta)$ \eqref{g}, the constraint \eqref{Virasoro bis k=1} evaluated along the locus $t=s=0$ gives
\begin{equation} \label{partial t_1}
\frac{\partial f}{\partial t_1}\Bigg|_{t=s=0}=\frac{1}{in}\,(e^{i\eta}-e^{i\theta})g(\eta,\theta).
\end{equation}
Consequently, along the locus $\eta=-\theta$, we have
\begin{equation*}
\frac{\partial f}{\partial t_1}\Bigg|_{\substack{t=s=0 \\ \eta=-\theta}}=-\frac{2}{n}\,\sin(\theta)R(\theta).
\end{equation*}

We then proceed by induction. We call
\begin{equation*}
\frac{\partial^nf}{\partial t_{j_1}\partial t_{j_2}\dots\partial t_{j_n}},
\end{equation*}
a $t$ derivative of weighted degree $|j|=j_1+j_2+\dots+j_n$. Then, for $k\geq 1$, we compute the system formed by
\begin{equation} \label{system}
\begin{cases}
\text{all $t$-derivatives of weighted degree $k$ of}\;\eqref{Virasoro bis k=1},\\
\text{all $t$-derivatives of weighted degree $k-1$ of}\;\eqref{Virasoro bis k=2},
\end{cases}
\end{equation}
evaluated at $t=s=0$. For instance, for $k=1$, \eqref{system} reduces to
\begin{align}
&\mathcal{B}_1(\eta,\theta)\Big(\frac{\partial f}{\partial t_1}\Big|_{t=s=0}\Big)=\frac{\partial f}{\partial t_{2}}\Bigg|_{t=s=0}+n\,\frac{\partial^2 f}{\partial t_1^2}\Bigg|_{t=s=0},\notag\\
&\mathcal{B}_2(\eta,\theta)f\Big|_{t=s=0}=\frac{\partial^2 f}{\partial t_1^2}\Bigg|_{t=s=0}+n\,\frac{\partial f}{\partial t_2}\Bigg|_{t=s=0}+\Bigg(\frac{\partial f}{\partial t_1}\Bigg|_{t=s=0}\Bigg)^2.\notag
\end{align}
After substitution of \eqref{partial t_1}, this system of equations can be solved for $\frac{\partial^2 f}{\partial t_1^2}\Big|_{t=s=0}$ and $\frac{\partial f}{\partial t_{2}}\Big|_{t=s=0}$ in terms of $\eta,\theta$, $g(\eta,\theta)$ and $\mathcal{D}g(\eta,\theta)$, whenever $n\neq 1$. Consequently, on the locus $\eta=-\theta$, the partials $\frac{\partial^2 f}{\partial t_1^2}\Big|_{\substack{t=s=0 \\ \eta=-\theta}}$ and $\frac{\partial f}{\partial t_{2}}\Big|_{\substack{t=s=0 \\ \eta=-\theta}}$ can be expressed in terms of $\theta$, $R(\theta)$ and $R'(\theta)$.\\

For general $k\geq 1$, suppose all the $t$-derivatives of $f$ of weighted degree $k$, evaluated at $t=s=0$, have been expressed in terms of $\eta,\theta$ and $g(\eta,\theta)$, \dots,$\mathcal{D}^{k-1}g(\eta,\theta)$, whenever $n\neq 1,\dots,k-1$. Then \eqref{system} is a system of linear equations where the unknowns are all the $t$-derivatives of $f$ of weighted degree $k+1$, evaluated at $t=s=0$, and the coefficients can be expressed in terms of $\eta,\theta$ and $g(\eta,\theta)$, \dots,$\mathcal{D}^{k-1}g(\eta,\theta)$. This is a system of $p(k)+p(k-1)$ linear equations in $p(k+1)$ unknowns, where $p(k)$ is the number of partitions of the natural number $k$. As $p(k+1)\leq p(k)+p(k-1)$, this system can be solved and all the $t$-derivatives of $f$ of weighted degree $k+1$, evaluated at $t=s=0$ can be expressed in terms of $\eta,\theta$, and $g(\eta,\theta)$, \dots, $\mathcal{D}^kg(\eta,\theta)$, whenever $n\neq k$. Consequently, on the locus $\eta=-\theta$, the $t$-derivatives of $f$ of weighted degree $k+1$, evaluated at $t=s=0$ and on the locus $\eta=-\theta$, can be expressed in terms of $\theta$, $R(\theta)$, $R'(\theta)$, \dots, $R^{(k)}(\theta)$.

Since the KP equation (\ref{KP}) contains $t$-derivatives of $f$ of weighted degree less or equal to 4, by performing the above scheme up to $k=3$, we can express all these derivatives, evaluated at $t=s=0$ and $\eta=-\theta$, in terms of $\theta$, $R(\theta)$ and its first three derivatives, whenever $n\geq 4$. This gives us a third order differential equation for $R(\theta)$:
\begin{multline*}
0=4R(\theta)^2-2\big(n^2+(1-n^2)\cos2\theta\big)R'(\theta)+8\sin2\theta\,R(\theta)R'(\theta)\\
-2\sin2\theta\,R''(\theta)+\sin^2\theta\big(12R'(\theta)^2-R'''(\theta)\big).
\end{multline*}
Multiplying the left-hand and the right-hand side of this equation with $\frac{1}{4}\,\sin\theta\,\Big(2\cos\theta\,R'(\theta)+\sin\theta\,R''(\theta)\Big)$, we obtain
\begin{equation} \label{TWbis}
0=\frac{\mbox{d}}{\mbox{d}\theta}\Big(\sin^2\theta\,R'(\theta)W(\theta)\Big),
\end{equation}
with
\begin{multline*}
W(\theta)=R(\theta)^2+2\sin\theta\,\cos\theta\,R(\theta)R'(\theta)+\sin^2\theta\,R'(\theta)^2\\
-\frac{1}{2}\Big(\frac{1}{4}\sin^2\theta\,\frac{R''(\theta)^2}{R'(\theta)}+\sin\theta\,\cos\theta\,R''(\theta)+\big(\cos^2\theta
+n^2\sin^2\theta\big)\,R'(\theta)\Big).
\end{multline*}
Equation \eqref{TWbis} implies that $W(\theta)=0$, which is the equation \eqref{TW}, obtained by Tracy and Widom in \cite{TW}. This concludes the proof of Theorem 3.1.
\end{proof}
\begin{remark} \emph{In the above proof, we had to assume that $n\geq 4$, where $n$ is the size of the random unitary matrices. For $n=1,2,3$, the function $R(\theta)$ also satisfies \eqref{TW}, as can be shown by direct computation, using the representation \eqref{taudet} of the probability $\tau_n(\eta,\theta)$ as a Toeplitz determinant. It would be interesting to relate the proof with the original derivation in \cite{TW}. For the Gaussian ensembles, the relation between the two methods has been studied in \cite{R}}.
\end{remark}
Finally, similarly to the case of the Jacobi polynomial ensemble (see \cite{HS}), we observe that $R(\theta)$ in \eqref{gaprob} is linked to the Painlevé VI equation. Precisely, we show that it is the restriction to the unit circle of a solution of (a special case of) the Painlevé VI equation, defined for $z\in\mathbb{C}$.
\begin{corollary}
Put $R(\theta)=r(e^{-2i\theta})$. Then, the function
\begin{equation*}
\sigma(z)=-i(z-1)r(z)-\frac{n^2}{4}z
\end{equation*}
satisfies the Okamoto-Jimbo-Miwa form of the Painlevé VI equation
\begin{multline}
[z(z-1)\sigma'']^2+4z(z-1)(\sigma')^3+4\sigma'\sigma^2+4(1-2z)\sigma(\sigma')^2\\-c_1(\sigma')^2
+[2(1-2z)c_4-c_2]\sigma'+4c_4\sigma-c_3=0,\label{PVI}
\end{multline}
with
\begin{equation} \label{constants}
c_1=n^2,\quad c_2=\frac{3n^4}{8},\quad c_3=\frac{n^6}{16},\quad c_4=-\frac{n^4}{16}.
\end{equation}
\end{corollary}
\begin{proof} From \eqref{TW}, by a straightforward computation, putting $R(\theta)=r(e^{-2i\theta})$, we obtain that $r(z)$ satisfies
\begin{multline}\label{cTW}
[z(z-1)r'']^2+4z^2(z-1)r'r''-4iz(z-1)^2(r')^3-4i(z^2-1)r(r')^2\\+[4z^2-n^2(z-1)^2](r')^2-4ir^2r'=0.
\end{multline}

Substituting in \eqref{cTW}
\begin{equation*}
r(z)=i\frac{\sigma(z)+xz}{z-1}
\end{equation*}
for some constant $x$, and annihilating the coefficient of $\sigma^2$, one finds that $x=n^2/4$.
With this choice of $x$, the new function $\sigma(z)$ satisfies the Painlevé VI equation (\ref{PVI}) if we pick $c_1,c_2,c_3$ and $c_4$ as in \eqref{constants}, which establishes Corollary 3.3.
\end{proof}
\section{Discussion of the results and some further directions} 
Our starting motivation was to understand a differential equation \eqref{TW} due to Tracy and Widom \cite{TW}, satisfied by the logarithmic derivative of the gap probability that an arc of circle of length $2\theta$ contains no eigenvalues of a randomly chosen unitary $n\times n$ matrix, from the point of view of the algebraic approach initiated by Adler, Shiota and van Moerbeke \cite{ASVM}. The main surprise is that the 2-dimensional Toda tau functions \eqref{tauint} deforming these gap probabilities, satisfy a centerless \emph{full} Virasoro algebra of constraints.
The result stands in contrast with  the corresponding integrals for the Gaussian or the orthogonal polynomial ensembles, which roughly satisfy only "half of" a Virasoro type algebra of constraints, see \cite{ASVM}, \cite{HS}, \cite{AVM1} and \cite{R}.

As mentioned at the beginning of Section 3, the integrals \eqref{tauint} can be expressed as Toeplitz determinants, see \eqref{taudet}. As such, they are very special instances of tau functions for the so-called Toeplitz lattices \cite{AVM2}, that is
\begin{equation}\label{Toeptau}
\tau_n(t,s)=\mbox{det}\big(\mu_{k-l}(t,s)\big)_{0\leq k,l\leq n-1},
\end{equation}
where
\begin{equation} \label{Toepmom}
\mu_k(t,s)=\int_{S^1} z^k\;e^{\sum_{j=1}^\infty(t_j z^j+s_j z^{-j})}w(z)\;\frac{\mbox{d}z }{2\pi i z},\quad k\in\mathbb{Z},
\end{equation}
and $w(z)$ is some (complex-valued) weight function defined on the unit circle $S^1$, such that the trigonometric moments
\begin{equation*}
\mu_k=\mu_k(0,0)=\int_{S^1} z^k w(z)\;\frac{\mbox{d}z }{2\pi i z},\quad k\in\mathbb{Z},
\end{equation*}
satisfy $\mbox{det}\big(\mu_{k-l}\big)_{0\leq k,l\leq n-1}\neq 0, \;\forall n\geq 1$. In the special case \eqref{taudet} that we consider in this paper, $w(z)=\chi_{(\eta,\theta)^c}(z)$ is the characteristic function of the complement of the arc of circle $(\eta, \theta)=\{z\in S^1\vert \eta<\mbox{arg}(z)< \theta\}$.

As it immediately follows from \eqref{Toepmom}, at the level of the trigonometric moments, the Toeplitz hierarchy is given by the simple equations
\begin{equation*}
T_j \mu_k\equiv\frac{\partial \mu_k}{\partial t_j}=\mu_{k+j},\qquad T_{-j}\mu_k\equiv\frac{\partial \mu_k}{\partial s_j}=\mu_{k-j},
\quad \forall j\geq 1.
\end{equation*}
Obviously $[T_i,T_j]=0, \forall i,j\in\mathbb{Z}$, if we define $T_0\mu_k=\mu_k$. Following an idea introduced in \cite{HS} in the context of the 1-dimensional Toda lattices, we define the following vector fields on the trigonometric moments
\begin{equation} \label{virmom}
V_j \mu_k=(k+j)\mu_{k+j},\quad \forall j\in\mathbb{Z}.
\end{equation}
These vector fields trivially satisfy the commutation relations
\begin{align}
[V_i,V_j]&=(j-i)V_{i+j}\label{m1}\\
[V_i,T_j]&=j T_{i+j},\quad \forall i,j\in\mathbb{Z},\label{m2}
\end{align}
from which it follows that
\begin{equation}\label{m3}
[[V_i,T_j],T_j]=j[T_{i+j},T_j]=0,\quad \forall i,j\in\mathbb{Z}.
\end{equation}
Equations \eqref{m1}, \eqref{m2} and \eqref{m3} mean that the vector fields $V_j, j\in\mathbb{Z}$, form a Virasoro algebra of master symmetries, in the sense of Fuchssteiner \cite{F}, for the Toeplitz hierarchy.

The tau functions \eqref{Toeptau} admit the following expansion
\begin{equation*}
\tau_n(t,s)=\sum_{\substack{0\leq i_0<\dots<i_{n-1} \\ 0\leq j_0<\dots<j_{n-1}}}p_{\substack{i_0,\dots,i_{n-1} \\ j_0,\dots,j_{n-1}}}S_{i_{n-1}-(n-1),\dots,i_0}(t)S_{j_{n-1}-(n-1),\dots,j_0}(s),
\end{equation*}
where
\begin{equation}\label{Plucker}
p_{\substack{i_0,\dots,i_{n-1} \\ j_0,\dots,j_{n-1}}}=\det\big(\mu_{i_k-j_l}(0,0)\big)_{0\leq k,l\leq n-1},
\end{equation}
are the so-called Plücker coordinates, and $S_{i_1,\dots,i_k}(t)$ denote the Schur polynomials
\begin{equation*}
S_{i_1,\dots,i_k}(t)=\det\big(S_{i_r+s-r}(t)\big)_{1\leq r,s\leq k},
\end{equation*}
with $S_n(t)$ the so-called elementary Schur polynomials defined by the generating function
\begin{equation*}
\exp\big(\sum_{k=1}^\infty t_kx^k\big)=\sum_{n\in\mathbb{Z}} S_n(t_1,t_2,\ldots) x^n.
\end{equation*}

In a forthcoming publication \cite{HV}, we shall establish the next result:
\begin{theorem} For all $k\in\mathbb{Z}$, we have
\begin{multline*}
L^{(n)}_k \tau_n(t,s)=\\\sum_{\substack{0\leq i_0<\dots<i_{n-1} \\ 0\leq j_0<\dots<j_{n-1}}}V_k \Big(p_{\substack{i_0,\dots,i_{n-1} \\ j_0,\dots,j_{n-1}}}\Big)S_{i_{n-1}-(n-1),\dots,i_0}(t)S_{j_{n-1}-(n-1),\dots,j_0}(s),
\end{multline*}
with $L^{(n)}_k, k\in \mathbb{Z}$, defined as in \eqref{virplus}, \eqref{virzero}, \eqref{virminus}, and $V_k \Big(p_{\substack{i_0,\dots,i_{n-1} \\ j_0,\dots,j_{n-1}}}\Big)$ the Lie derivative of the Plücker coordinates \eqref{Plucker} in the direction of the master symmetries $V_k$ of the Toeplitz hierarchy, as defined in \eqref{virmom}.
\end{theorem}

Thus the operators $L^{(n)}_k, k\in \mathbb{Z}$, precisely describe the master symmetries of the Toeplitz hierarchy on the tau functions of this hierarchy. Since master symmetries are usually connected with a bi-hamiltonian structure in the sense of Magri \cite{Ma} (see \cite{MZ} and \cite{D} for an overview), it suggests to investigate the existence of a bi-hamiltonian structure for the Toeplitz hierarchy, which seems to be an open problem.

\end{document}